\documentclass[conference]{IEEEtran}
\IEEEoverridecommandlockouts

\usepackage{xcolor}
\usepackage{dsfont}

\def\BibTeX{{\rm B\kern-.05em{\sc i\kern-.025em b}\kern-.08em
    T\kern-.1667em\lower.7ex\hbox{E}\kern-.125emX}}

\usepackage{cite}
\ifCLASSINFOpdf
  \usepackage[pdftex]{graphicx}
  \usepackage{amsfonts}
  \usepackage{xcolor}

  \usepackage{bbm}
  \usepackage{graphicx}
  \usepackage{makecell}
  \usepackage{hhline}
  \graphicspath{{figures/}}

  \usepackage{booktabs} 
  \usepackage{array}
\usepackage[caption=false,font=footnotesize]{subfig}
  \usepackage{amsmath,amssymb,amsthm}
  \usepackage[colorlinks=true, linkcolor=black, citecolor=black, urlcolor=black]{hyperref}
  
\usepackage{etoolbox}

\makeatletter
\patchcmd{\@makecaption}
  {\ifCLASSOPTIONconference
   \hbox to\hsize{\normalfont\footnotesize\hfil\box\@tempboxa\hfil}}
  {\ifCLASSOPTIONconference
   \hbox to\hsize{\normalfont\footnotesize\box\@tempboxa\hfil}}
  {}{}
\makeatother

\makeatletter

\renewcommand{\section}{%
\@startsection{section}{1}{\z@}%
{1.4ex plus 0.4ex minus 0.2ex}% 标题前
{0.6ex plus 0.2ex minus 0.1ex}% 标题后
{\normalfont\normalsize\centering\scshape}}

\renewcommand{\subsection}{%
\@startsection{subsection}{2}{\z@}%
{1.0ex plus 0.3ex minus 0.2ex}% 标题前
{0.4ex plus 0.2ex minus 0.1ex}% 标题后
{\normalfont\normalsize\itshape}}

\makeatother

\AtBeginDocument{%
  \setlength{\abovedisplayskip}{3pt plus 1pt minus 1pt}
  \setlength{\belowdisplayskip}{3pt plus 1pt minus 1pt}
  \setlength{\abovedisplayshortskip}{2pt plus 1pt minus 1pt}
  \setlength{\belowdisplayshortskip}{2pt plus 1pt minus 1pt}
  \setlength{\jot}{2pt} % 多行公式内部行距

  \setlength{\abovecaptionskip}{1pt}
  \setlength{\belowcaptionskip}{0pt}

  \setlength{\floatsep}{4pt plus 1pt minus 1pt}
  \setlength{\textfloatsep}{6pt plus 1pt minus 1pt}
  \setlength{\intextsep}{6pt plus 1pt minus 1pt}

  \setlength{\dblfloatsep}{4pt plus 1pt minus 3pt}
  \setlength{\dbltextfloatsep}{6pt plus 1pt minus 1pt}

  \setlength{\dblfloatsep}{0pt} % 两个 figure* 之间
  \captionsetup[subfloat]{captionskip=0pt}
}

  \newtheorem{proposition}{Proposition}[section]

\newtheorem{lemma}{Lemma}[section]

\else
\fi

\begin{document}

\bstctlcite{IEEEexample:BSTcontrol}
\title{Reliable Narrowband Interference Detection via Backward Conformal Prediction}

\author{
	\IEEEauthorblockN{Xin~Su\textsuperscript{*}, Meiyi~Zhu\textsuperscript{\dag}, Osvaldo~Simeone\textsuperscript{\ddag}, Marco~Di~Renzo\textsuperscript{\S,\dag}, and Carlo~Fischione\textsuperscript{*}}
	\IEEEauthorblockA{\textsuperscript{*}KTH Royal Institute of Technology, Stockholm, Sweden}
	\IEEEauthorblockA{\textsuperscript{\dag}King's College London, London, UK}
	\IEEEauthorblockA{\textsuperscript{\ddag}Northeastern University London, London, UK}
    \IEEEauthorblockA{\textsuperscript{\S}CNRS and CentraleSup\'elec, Institute of Electronics and Digital Technologies, Cesson-S\'evign\'e, France}
	\vspace{-32pt} 
    
\thanks{X. Su and C. Fischione are with the School of Electrical Engineering and Computer Science, KTH Royal Institute of Technology, 10044 Stockholm, Sweden (e-mail: xisu@kth.se; carlofi@kth.se).
M. Zhu is with the Department of Engineering, King's College London, London WC2R 2LS, U.K. (e-mail: meiyi.1.zhu@kcl.ac.uk).
O. Simeone is with the Institute for Intelligent Networked Systems, Northeastern University London, London E1 8PH, U.K. (e-mail: o.simeone@northeastern.edu).
M. Di Renzo is with CNRS and CentraleSup\'elec, Institute of Electronics and Digital Technologies (IETR), Avenue de la Boulaie, 35576 Cesson-S\'evign\'e, France (marco.direnzo@centralesupelec.fr), and with King's College London, Department of Engineering-Centre for Telecommunications Research, WC2R 2LS London, United Kingdom (marco.di\_renzo@kcl.ac.uk). The work of X. Su and C. Fischione was sponsored by the KTH DF research center and by the SSF SAICOM project. The work of M. Zhu and O. Simeone was supported by an Open Fellowship of the EPSRC (EP/W024101/1). The work of O. Simeone was also supported by EPSRC (EP/X011852/1) and ERC (No. 101198347).
% The work of X. Su and C. Fischione was sponsored by the KTH DF research center and by the SSF SAICOM project. The work of O. Simeone was supported by the European Research Council (ERC) under the European Union’s Horizon Europe Programme (grant agreement No. 101198347), by an Open Fellowship of the EPSRC (EP/W024101/1), and by the EPSRC project (EP/X011852/1).
}
}

\maketitle

\begin{abstract}
Narrowband interference can severely degrade the performance of wireless links by concentrating significant power on a small portion of the channel. Machine learning detectors trained on baseband I/Q samples can identify the affected subcarriers with high accuracy, surpassing model-based detectors that rely on hand-crafted statistics. The predictive probabilities produced by such detectors are, however, typically poorly calibrated, and downstream mitigation modules generally operate under strict resource budgets that limit the number of candidate interference subcarriers that can be acted upon. Conformal prediction (CP) provides a distribution-free framework for constructing prediction sets that control the miscoverage level, i.e., the probability of excluding the true output, at a prescribed level. However, this target miscoverage level must be fixed in advance, while the resulting prediction-set size remains uncontrolled, which is misaligned with operationally constrained settings. To address this issue, we develop a backward conformal prediction (BCP) framework in which the prediction-set size is fixed by the operational budget and the corresponding per-input miscoverage level is estimated from calibration data with a provable post-hoc reliability guarantee. 
% We instantiate the framework for narrowband interference detection in wireless  systems and show through simulations that BCP yields reliable miscoverage estimates whose accuracy approaches that of an uncalibrated baseline with sufficient calibration data.
We instantiate the framework for narrowband interference detection in wireless
systems and show through simulations that BCP yields reliable miscoverage
estimates with accuracy comparable to that of both the uncalibrated and
temperature-scaled baselines.
\end{abstract}


\IEEEpeerreviewmaketitle

\section{Introduction}
\label{sec:intro}

Narrowband interference (NBI) is a pervasive impairment in wireless communications. Although localized in the frequency domain, NBI concentrates significant power on a small number of subcarriers and can substantially degrade legitimate wireless  links~\cite{aygur2025narrowband}. 
Identifying the subcarrier affected by NBI is therefore a key prerequisite for effective mitigation.
% Detecting whether NBI is present, and, if so, identifying which subcarriers are affected, is therefore a key prerequisite for effective mitigation.

Conventional NBI detectors rely on energy or correlation statistics and require accurate models of the legitimate signal and of the interferer~\cite{hadaschik2007joint,gonzalez2005narrowband}, which are difficult to obtain in practice under unknown waveforms and time-varying channels. Machine learning (ML) detectors offer a flexible alternative: by learning discriminative features directly from the received measurements, e.g., in-phase and quadrature (I/Q) samples, deep models achieve high detection and classification accuracy across a wide range of operating conditions~\cite{robinson2023narrowband,andersson2024deep,hu2024narrowband}, with recent work further improving robustness to distributional mismatch~\cite{xiao2025robustness}, addressing low-power interference~\cite{jia2025low}, and enhancing interpretability through explainable AI~\cite{hinkley2025quantifying}. 
% ML-based wireless components are now increasingly considered for deployment in real systems~\cite{letaief2019roadmap}.

\begin{figure}[t]
    \centering
    \includegraphics[width=\linewidth]{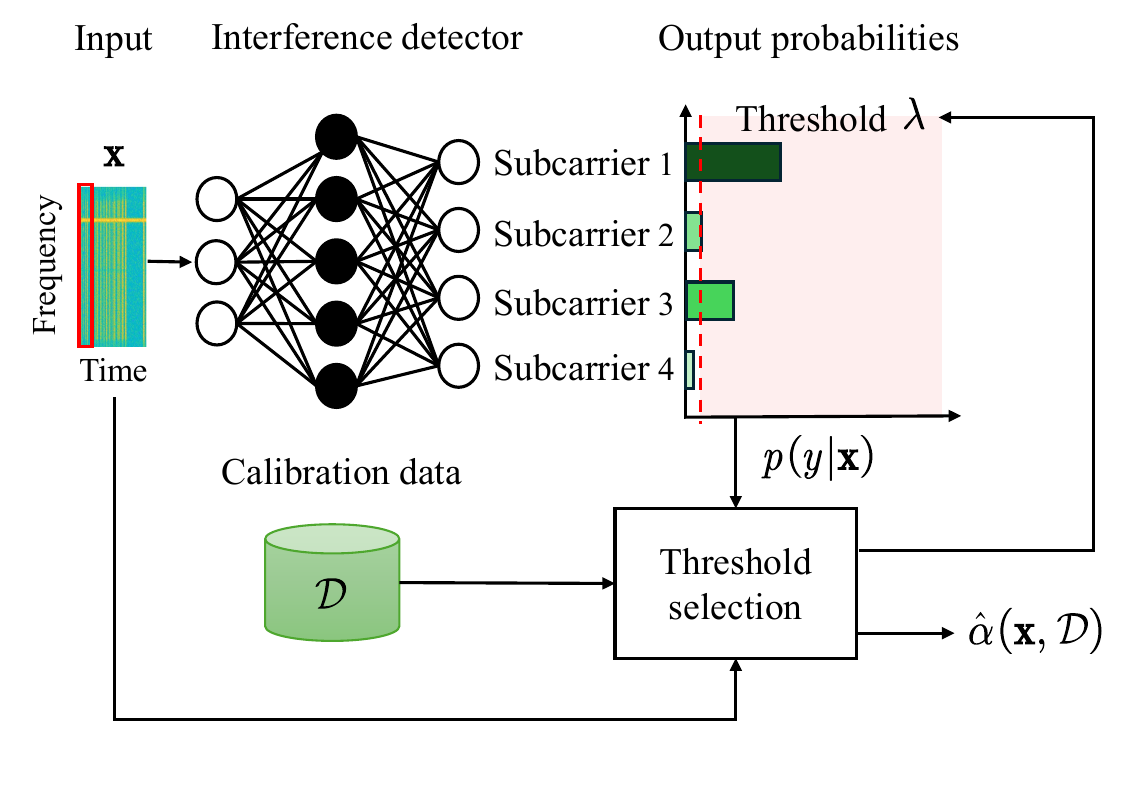}
    \caption{A wireless link may be affected by narrowband interference on at most one subcarrier~\cite{robinson2023narrowband,grover2014jamming}. A pre-trained model processes the received I/Q data $\mathbf{x}$ and outputs predictive probabilities $\hat{p} (y|\mathbf{x})$ over the candidate labels $y$, which consist of $S$ candidate subcarriers possibly affected by interference. Using a calibration dataset $\mathcal{D}$, for a predefined top-$K$ set predictor, we wish to produce a reliable estimate $\hat{\alpha}(\mathbf{x},\mathcal{D})$ of the probability that the set does not include the true label.}
    \label{fig:set_id}
\end{figure}

Beyond raw accuracy, deploying ML detectors in quality- or safety-critical settings calls for an assessment of the \emph{reliability} of their outputs. A natural approach is to leverage directly the self-reported softmax probabilities of the underlying model. As is well known, however, modern deep neural networks are typically poorly calibrated, and overconfidence is especially common~\cite{guo2017calibration,zecchin2023robust}. \emph{Conformal prediction} (CP) provides a principled remedy by returning prediction sets with distribution-free guarantees on the miscoverage level, i.e., on the probability of excluding the true label~\cite{vovk2005algorithmic}. This property makes CP particularly attractive for wireless systems, where accurate statistical modeling is generally elusive, and recent works have leveraged CP for tasks such as demodulation~\cite{cohen2023calibratingj}, beamforming~\cite{su2025conformal}, distributed inference~\cite{zhu2024federated}, and network performance analysis~\cite{hou2025what}.

A key limitation of standard CP is that one must prescribe the target miscoverage level of the prediction set in advance, while the resulting set size is left uncontrolled. A lower target miscoverage level generally improves reliability by including more candidate subcarriers, thereby increasing the set size. This is poorly aligned with many wireless deployments, where downstream modules can act on only a limited number of candidate subcarriers, dictated by operational constraints such as compute, memory, or latency budgets. In such settings, the question of interest is reversed: given a prediction set whose size is fixed by an operational budget, what is the probability that it fails to contain the true label? 

The recently introduced \emph{backward conformal prediction} (BCP) framework provides a principled approach to this question. Building on e-values as test statistics~\cite{chugg2026values,balinsky2024enhancing}, BCP fixes the set rule and produces a per-input estimate of the corresponding miscoverage level with a provable post-hoc reliability guarantee~\cite{gauthier2025backward,gauthier2025values}.

In this paper, we leverage BCP to assess the reliability of operationally constrained prediction sets for NBI detection. Unlike prior work on NBI detection, which focuses on improving classification accuracy or robustness~\cite{robinson2023narrowband,andersson2024deep,hu2024narrowband,xiao2025robustness,jia2025low,hinkley2025quantifying}, we take a pre-trained probabilistic detector as given and quantify the risk that its budget-constrained output excludes the true interference state. The main contributions are as follows:

\begin{itemize}

    \item We formulate an operationally constrained, generalized top-$K$ prediction-set framework for NBI detection, in which different candidate subcarriers may carry different mitigation costs.

    \item We develop a BCP-based per-input miscoverage estimator for the resulting constrained prediction sets in closed form and establish its distribution-free reliability guarantee.

    % \item Through simulations on IEEE~802.11a/g signals, we show that the proposed estimator is conservative across a range of operating conditions, while its Brier score approaches that of an uncalibrated baseline with sufficient calibration data.
    \item Through simulations on IEEE 802.11a/g signals, we show that the proposed estimator provides reliable miscoverage estimates across a range of operating conditions, while achieving high estimation accuracy.
\end{itemize}

% The remainder of the paper is organized as follows. Section~\ref{sec:prob_def} formulates the constrained prediction-set problem for NBI detection. Section~\ref{sec:bcp_proc} introduces a na\"{\i}ve baseline and develops the BCP-based miscoverage estimator. Section~\ref{sec:numerical_results} reports numerical results, and Section~\ref{sec:conclusion} concludes the paper.

\section{System Model and Problem Definition}
\label{sec:prob_def}

%-------------------------------------------------------------
\subsection{System Model}
\label{subsec:sys_model}

We consider a wideband wireless system operating in the presence of NBI~\cite{robinson2023narrowband}. Following~\cite{robinson2023narrowband}, the receiver collects baseband I/Q samples over an observation window of length $M$, yielding the input vector $\mathbf{x}\in\mathbb{C}^{M}$. Based on $\mathbf{x}$, the receiver must determine whether NBI is present and, if so, identify the affected subcarrier. 
We focus on NBI-affected observations and monitor $S$ subcarriers, assuming, as in~\cite{robinson2023narrowband}, that the NBI affects only one of them. The output $y$ identifies the affected subcarrier, with label space $\mathcal{Y}=\{1,\ldots,S\}$.

A pre-trained probabilistic detector maps the input $\mathbf{x}$ to a predictive distribution $\hat{p} (y|\mathbf{x})$ over $\mathcal{Y}$. In a conventional point-estimation approach, the receiver typically applies the maximum a posteriori (MAP) rule and returns the most likely label
\begin{align}
\label{eq:point_est}
\hat{y}=\arg\max_{y\in\mathcal{Y}}\;\hat{p} (y|\mathbf{x}).
\end{align}
However, the detector's own confidence values are often poorly calibrated~\cite{zecchin2023robust} tending to underestimate the true miscoverage probability~\cite{guo2017calibration}. Therefore, the MAP estimate in~\eqref{eq:point_est} is not guaranteed to minimize the error rate, and the estimate
\begin{equation}
\label{eq:map_misclassification}
    1-\max_{y\in\mathcal{Y}}\hat{p}(y|\mathbf{x})
\end{equation}
often underestimate the true error probability $\Pr(\hat{y}\neq y)$.

As discussed in Sec.~\ref{sec:intro}, beyond point decisions, downstream mitigation may benefit from retaining multiple plausible interference subcarriers. 
Since only a limited number of subcarriers can be acted upon under the available mitigation budget, we next extend the point estimates to generalized top-$K$ prediction sets~\cite{cohen2023calibratingj}.

\begin{figure*}[t]
    \centering
    \includegraphics[width=0.9\linewidth]{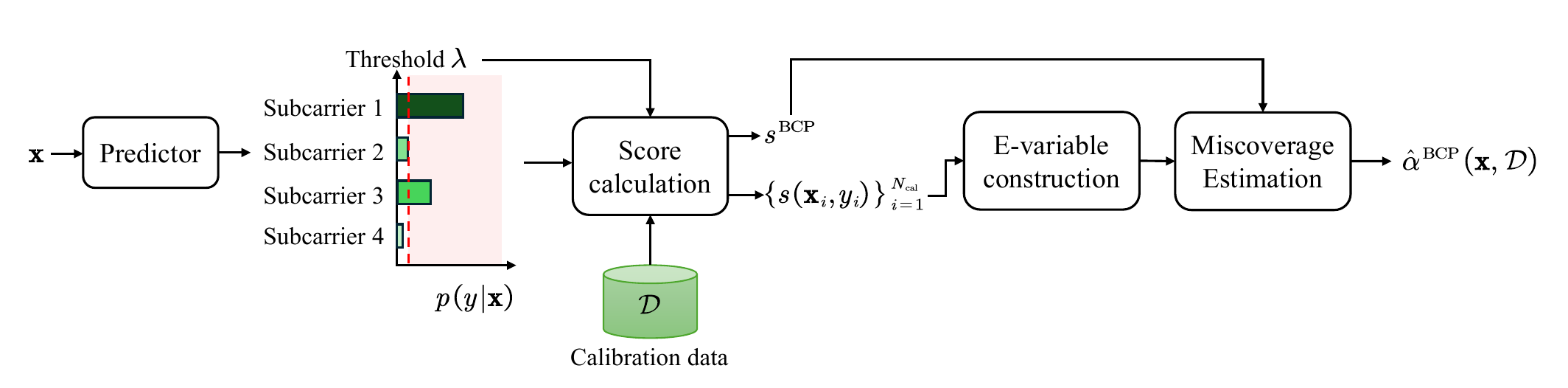}
    \caption{Given the predictive probabilities $\hat{p} (y|\mathbf{x})$ and the calibration data $\mathcal{D}$, BCP defines a nonconformity score $s(\mathbf{x},y)$ and produces a miscoverage estimate $\hat{\alpha}^{\mathrm{BCP}}(\mathbf{x},\mathcal{D})$ via~\eqref{eq:alpha_hat_closed_form}, which satisfies the reliability condition in~\eqref{eq:cvg_guar_e}.}
    \label{fig:bcp_proc}
\end{figure*}

\subsection{Generalized Top-\textit{K} Prediction Sets}
\label{subsec:pred_set}

As illustrated in Fig.~\ref{fig:set_id}, given a predictive distribution $\hat{p} (y|\mathbf{x})$, a prediction set is typically formed by retaining all labels whose predictive probabilities exceed a threshold $\lambda\in(0,1)$, i.e.,
\begin{align}
\label{eq:naive_threshold_set}
\mathcal{C}_{\lambda}(\mathbf{x})= \bigl\{ y\in\mathcal{Y}:\hat{p}(y| \mathbf{x})> \lambda \bigr\}.
\end{align}
The set $\mathcal{C}_{\lambda}(\mathbf{x})$ guides downstream NBI mitigation  mechanisms such as blocking the use of all subcarriers in this set~\cite{aygur2025narrowband}. Accordingly, in general, the inclusion of subcarrier $y$ into set $\mathcal{C}_{\lambda}(\mathbf{x})$ comes with a cost $c_y\in[0,1]$.
For example, fifth-generation New Radio (5G NR) supports licensed-assisted access, where unlicensed subcarriers supplements licensed primary subcarriers through bandwidth aggregation~\cite{3gpp38300}. Suppressing the use of licensed subcarriers may incur a higher cost $c_y$ than blocking an unlicensed subcarrier.

For input $\mathbf{x}$, order the labels as $\hat{p
}(y^{(1)}|\mathbf{x})\ge \cdots\ge \hat{p
}(y^{(S)}|\mathbf{x})$, with ties broken arbitrarily. A generalized top-$K$ set predictor is obtained by choosing the smallest threshold in~\eqref{eq:naive_threshold_set} that respects an operational budget $K$, namely
\begin{equation}
  \label{eq:threshold}
\lambda^{\star}(\mathbf{x})=\underset{\lambda}{\mathrm{arg\,min}}\left\{ \lambda:\left| \mathcal{C}_{\lambda}(\mathbf{x})\right|\le C_{\max}(\mathbf{x})\right\},    
\end{equation}
where the budget-feasible maximum set size is
\begin{equation}
    C_{\max}(\mathbf{x})\!=\!\max\left\{\!m\in\!\{0,1,\ldots,S\}\!:\!\sum_{s=1}^{m} c_{y^{(s)}}\le K\!\right\}.
    \label{eq:Cmax_intf}
\end{equation}
In words, $C_{\max}(\mathbf{x})$ in \eqref{eq:Cmax_intf} is the maximum number of most likely labels that can be retained within the mitigation budget~$K$.

\subsection{Problem Statement}
\label{subsec:prob_stmt} 

Let $(\mathbf{x},y)$ denote the test pair, where $\mathbf{x}$ is the received I/Q vector during inference and $y$ is the corresponding unknown true label. We assume the availability of a calibration dataset $\mathcal{D}=\{(\mathbf{x}_i,y_i)\}_{i=1}^{N_{\mathrm{cal}}}$, where each pair consists of a received I/Q vector $\mathbf{x}_i$ and the corresponding interference label $y_i$. The calibration dataset $\mathcal{D}$ and the test pair $(\mathbf{x},y)$ are drawn i.i.d.\ from the same unknown distribution $P(\mathbf{x},y)$.

For notational simplicity, we denote the prediction set induced by
\eqref{eq:threshold} for the test input $\mathbf{x}$ as
$\mathcal{C}(\mathbf{x})\triangleq
\mathcal{C}_{\lambda^{\star}(\mathbf{x})}(\mathbf{x})$. 
A miscoverage event occurs when the prediction set does not contain the true label, i.e., $y\notin\mathcal{C}(\mathbf{x})$.
This event is represented by the indicator
$\mathds{1}(y\notin\mathcal{C}(\mathbf{x}))$, which is not directly
observable during inference since the label $y$ is unknown.
We consider only nontrivial budgets $K$ for which the set-size limit in~\eqref{eq:Cmax_intf} satisfies the inclusion $C_{\max}(\mathbf{x})\in\{1,\ldots,S-1\}$, thus excluding the empty- and full-set cases.

Specifically, to support reliable downstream decision-making, the estimate $\hat{\alpha}(\mathbf{x},\mathcal{D})$ should provide a conservative assessment of the true miscoverage. Following~\cite{koning2023post}, we formalize this requirement through the reliability condition
\begin{align}
\label{eq:cvg_guar_e}
\mathbb{E}
\left[
\frac{
\mathbb{P}\left(
y\notin\mathcal{C}(\mathbf{x})\,|\,\mathbf{x}
\right)
}{
\hat{\alpha}\left(\mathbf{x},\mathcal{D}\right)
}
\right]
=
\mathbb{E}
\left[
\frac{
\mathds{1}\left(
y\notin\mathcal{C}(\mathbf{x})
\right)
}{
\hat{\alpha}\left(\mathbf{x},\mathcal{D}\right)
}
\right]
\le 1.
\end{align}
where the first expectation is taken over both the calibration data $\mathcal D$
and the test input $\mathbf{x}$, while the second expectation also averages over the test label $y$. Condition~\eqref{eq:cvg_guar_e} penalizes underestimation: a smaller $\hat{\alpha}(\mathbf{x},\mathcal{D})$ incurs a larger penalty whenever a miscoverage event occurs.

The ratio guarantee in~\eqref{eq:cvg_guar_e} also yields the interpretable bound on the marginal probability~\cite{gauthier2025backward,gauthier2025values}
\begin{align}
\label{eq:cvg_guar_bound}
\Pr\!\left(y\notin\mathcal{C}(\mathbf{x})\right)
\leq
\mathbb{E}\!\left[\hat{\alpha}(\mathbf{x},\mathcal{D})\right]
+
\frac{
\sqrt{\operatorname{Var}\!\left(
\hat{\alpha}(\mathbf{x},\mathcal{D})\right)}
}{
2\mathbb{E}\!\left[\hat{\alpha}(\mathbf{x},\mathcal{D})\right]
},
\end{align}
where the probability, expectation, and variance are taken over the calibration data $\mathcal{D}$ and the test pair $(\mathbf{x},y)$. The derivation
of~\eqref{eq:cvg_guar_bound} is provided in the Appendix.
When the estimate $\hat{\alpha}(\mathbf{x},\mathcal{D})$ is sufficiently concentrated around its mean, the second term in~\eqref{eq:cvg_guar_bound} becomes negligible, yielding the approximate inequality
\begin{align}
\Pr\!\left(y\notin\mathcal{C}(\mathbf{x})\right)
\lesssim
\mathbb{E}\!\left[\hat{\alpha}(\mathbf{x},\mathcal{D})\right].
\label{eq:cvg_guar}
\end{align}

\section{Backward Conformal Prediction}
\label{sec:bcp_proc}
In this section, we first present a na\"{\i}ve miscoverage estimator (NME) that uses only the detector's confidence values. We then introduce the BCP-based estimator~\cite{gauthier2025backward}, which provably satisfies~\eqref{eq:cvg_guar_e}. The overall procedure is illustrated in Fig.~\ref{fig:bcp_proc}.

\subsection{Na\"{\i}ve Miscoverage Estimate}
\label{subsec:naive}
A straightforward estimate of the miscoverage level can be obtained directly from the detector's confidence values $\hat{p}(y|\mathbf{x})$, without any calibration data, by extending the model-based MAP misclassification probability in~\eqref{eq:map_misclassification} to the generalized top-$K$ setting. Specifically, the na\"{\i}ve miscoverage estimate (NME) computes the probability mass outside the prediction set $\mathcal{C}(\mathbf{x})$ as
\begin{align}
    \hat{\alpha}^{\mathrm{NME}}(\mathbf{x})
    =
    1-\sum_{y \in \mathcal{C}(\mathbf{x})} \hat{p}(y|\mathbf{x}).
    \label{eq:naive_alpha_discrete}
\end{align}
As with the MAP misclassification estimate in~\eqref{eq:map_misclassification}, the NME is generally not guaranteed to satisfy~\eqref{eq:cvg_guar_e}, owing to the miscalibration of the posterior $\hat{p}(y|\mathbf{x})$.

%-------------------------------------------------------------

\subsection{Backward Conformal Prediction}
\label{subsec:bcp}

%--------------------------------------------------------------------
\begin{figure*}[t]
    \centering

    \subfloat[Average ($\mathrm{SIR}=0$ dB).\label{fig:miscvg_func_0db}]{%
        \parbox[t]{0.185\textwidth}{%
            \centering
            \includegraphics[
                width=\linewidth,
                height=2.85cm,
                keepaspectratio
            ]{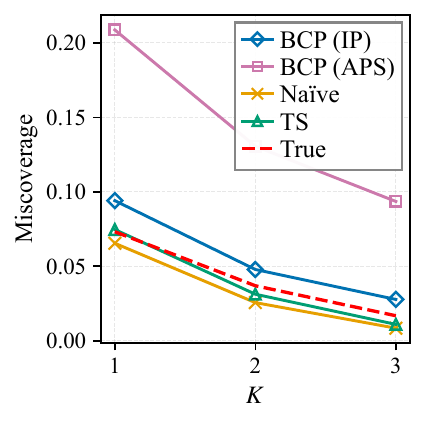}%
        }%
    }%
    \hspace{0.003\textwidth}%
    \subfloat[Distribution ($\mathrm{SIR}=0$ dB).\label{fig:miscvg_violin_0db}]{%
        \parbox[t]{0.31\textwidth}{%
            \centering
            \raisebox{0.45mm}{%
                \includegraphics[
                    width=\linewidth,
                    height=3.2cm,
                    keepaspectratio
                ]{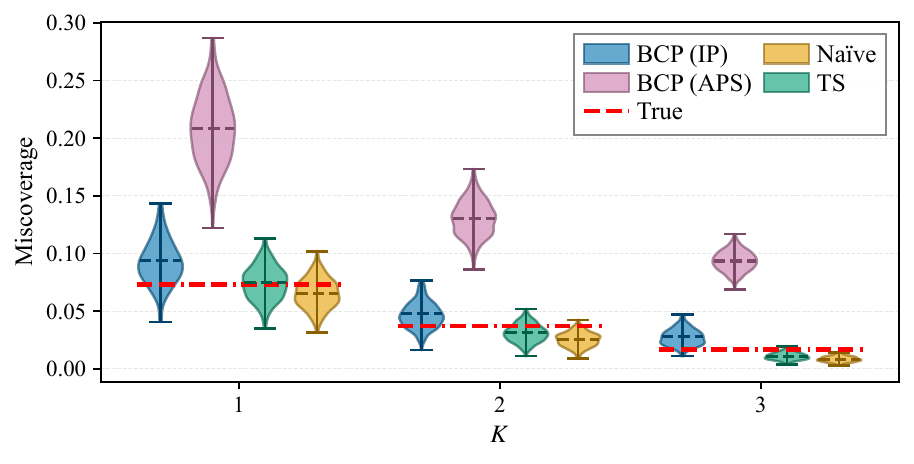}%
            }%
        }%
    }%
    \hspace{0.003\textwidth}%
    \subfloat[Average ($\mathrm{SIR}=5$ dB).\label{fig:miscvg_func_5db}]{%
        \parbox[t]{0.185\textwidth}{%
            \centering
            \includegraphics[
                width=\linewidth,
                height=2.85cm,
                keepaspectratio
            ]{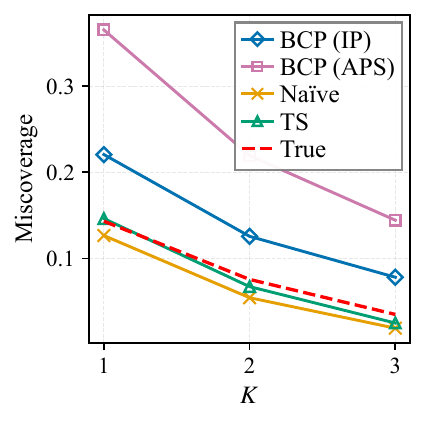}%
        }%
    }%
    \hspace{0.003\textwidth}%
    \subfloat[Distribution ($\mathrm{SIR}=5$ dB).\label{fig:miscvg_violin_5db}]{%
        \parbox[t]{0.31\textwidth}{%
            \centering
            \raisebox{0.45mm}{%
                \includegraphics[
                    width=\linewidth,
                    height=3.2cm,
                    keepaspectratio
                ]{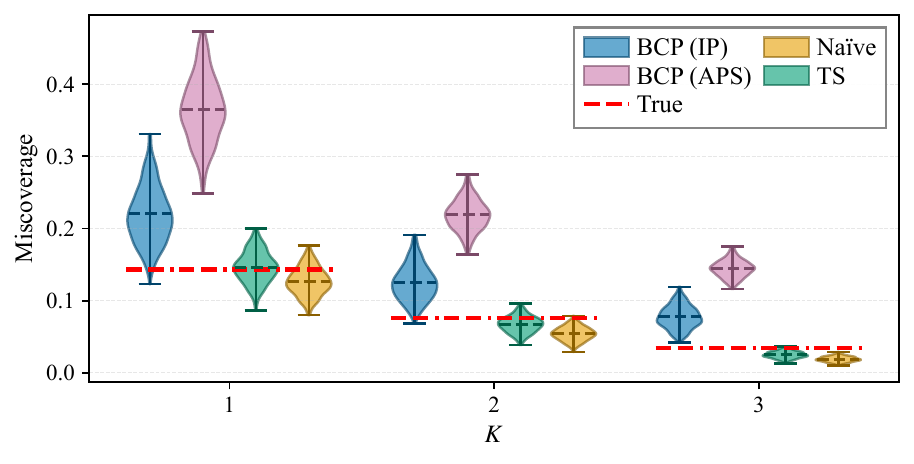}%
            }%
        }%
    }

    \caption{Average estimated miscoverage rates~\eqref{eq:metric_emr}
    and corresponding distributions at $\mathrm{SIR}\in\{0,5\}$ dB.
    The averages are plotted versus the budget $K$, while the violin
    plots show the distributions over the $N_{\mathrm{run}}$ experiments. 
    % The dashed lines inside the violins indicate the corresponding average.
    The red dashed lines indicate the average true miscoverage
    rates~\eqref{eq:metric_tmr}.}
    \label{fig:miscoverage_results}
\end{figure*}
%--------------------------------------------------------------------

To quantify how well a label $y$ conforms to the detector output, we define the nonconformity (NC) score via the inverse-power (IP) transformation
\begin{align}
\label{eq:nc_score_power}
s(\mathbf{x},y)=\frac{1}{\hat{p} (y|\mathbf{x})^{\beta}},
\end{align}
where $\beta>0$ is a hyperparameter. A larger score indicates a less likely label, since $s(\mathbf{x},y)$ is monotonically decreasing in $\hat{p} (y|\mathbf{x})$. Other non-increasing functions of $\hat{p} (y|\mathbf{x})$ could also be used.

For a test input $\mathbf{x}$ and a candidate label $y$, BCP forms the e-variable~\cite{balinsky2024enhancing,gauthier2025backward}
\begin{align}
    E(\mathbf{x},y) =
    \frac{s(\mathbf{x},y)}
    {\frac{1}{N_{\mathrm{cal}}+1}\left(\sum_{i=1}^{N_{\mathrm{cal}}} s(\mathbf{x}_i,y_i)+s(\mathbf{x},y)\right)}.
    \label{eq:e_value}
\end{align}
This quantity is the ratio between the test score $s(\mathbf{x},y)$ and the average score formed by combining this test score with the calibration scores $\{s(\mathbf{x}_i,y_i)\}_{i=1}^{N_{\mathrm{cal}}}$. We refer to~\cite{chugg2026values} for a general definition of e-variables and a discussion of their use as test statistics. Since $E(\mathbf{x},y)$ is increasing in $s(\mathbf{x},y)$ and $s(\mathbf{x},y)$ is decreasing in $\hat{p} (y|\mathbf{x})$, a less likely candidate label leads to a larger e-value.

We next use the e-variable in~\eqref{eq:e_value} to derive the BCP miscoverage estimate for $\mathcal{C}(\mathbf{x})$. As explained in Sec.~\ref{subsec:prob_stmt}, we consider the nontrivial case $C_{\max}(\mathbf{x})\in\{1,\ldots,S-1\}$ and first determine the probability threshold induced by the budget constraint.

\begin{lemma}
\label{lemma:threshold_selection}
For the test input $\mathbf{x}$, the threshold $\lambda^{\star}(\mathbf{x})$ selected by~\eqref{eq:threshold} is given by
\begin{equation}
\label{eq:lambda0_star}
    \lambda^{\star}(\mathbf{x})
    =
    \hat{p}\left(y^{(C_{\max}(\mathbf{x})+1)}|\mathbf{x}\right).
\end{equation}
\end{lemma}

\begin{IEEEproof}
By~\eqref{eq:Cmax_intf}, the budgeted prediction set contains the first $C_{\max}(\mathbf{x})$ labels in decreasing order of $\hat{p} (y|\mathbf{x})$. Hence, the first excluded label is $y^{(C_{\max}(\mathbf{x})+1)}$. Since~\eqref{eq:naive_threshold_set} keeps labels with predictive probability above the threshold, the boundary threshold is given by~\eqref{eq:lambda0_star}.
\end{IEEEproof}

Combining the threshold in~\eqref{eq:lambda0_star} with the e-variable construction in~\eqref{eq:e_value} gives the closed-form BCP miscoverage estimate and establishes its reliability guarantee.

\begin{proposition}[Closed-Form BCP miscoverage estimate]
\label{prop:bcp_closed_form_validity}
For the generalized top-$K$ prediction set $\mathcal{C} (\mathbf{x})$, the BCP miscoverage estimate is
\begin{align}
\label{eq:alpha_hat_closed_form}
\hat{\alpha}^{\mathrm{BCP}}(\mathbf{x},\mathcal{D})
=
\frac{1}{E\left(\mathbf{x},y^{(C_{\max}(\mathbf{x})+1)}\right)},
\end{align}
which satisfies the reliability condition in \eqref{eq:cvg_guar_e}, i.e.,
\begin{align}
\label{eq:post_hoc}
    \mathbb{E}\left[
    \frac{\mathds{1}\bigl(y\notin \mathcal{C} (\mathbf{x})\bigr)}
    {\hat{\alpha}^{\mathrm{BCP}}(\mathbf{x},\mathcal{D})}
    \right]
    \le 1.
\end{align}

\end{proposition}

\begin{IEEEproof}
Since $E(\mathbf{x},y)$ is decreasing in $\hat{p} (y|\mathbf{x})$, the probability ordering
$\hat{p}(y^{(1)}|\mathbf{x})\ge \cdots \ge \hat{p}(y^{(S)}|\mathbf{x})$ induces the e-value ordering $E(\mathbf{x},y^{(1)})\le  \cdots \le E(\mathbf{x},y^{(S)})$. 
Hence, by Lemma~\ref{lemma:threshold_selection}, the generalized top-$K$
prediction set defined in~\eqref{eq:naive_threshold_set} can be equivalently written as
\begin{equation}
    \mathcal{C} (\mathbf{x})
    =
    \left\{
    y\in\mathcal{Y}:
    E(\mathbf{x},y)
    <
    E\left(\mathbf{x},y^{(C_{\max}(\mathbf{x})+1)}\right)
    \right\}.
    \label{eq:set_e_value}
\end{equation}
Following~\cite{gauthier2025backward,gauthier2025values}, $\hat{\alpha}^{\mathrm{BCP}}(\mathbf{x},\mathcal{D})$ is obtained by solving 
\begin{equation}
\label{eq:alpha_hat_implicit}
\begin{aligned}
\hat{\alpha}^{\mathrm{BCP}}\!(\mathbf{x},\mathcal{D})
\!=\!
\inf_{\alpha >0}
&\; \alpha \\
\mathrm{s.t.}
&\;
\left|\left\{y\in\mathcal{Y}\!:\!E(\mathbf{x},y)\!<\!\frac{1}{\alpha}\right\}\right|
\!\le\! C_{\max}(\mathbf{x}).
\end{aligned}
\end{equation}
Using the threshold e-value $E\left(\mathbf{x},y^{(C_{\max}(\mathbf{x})+1)}\right)$ in~\eqref{eq:set_e_value}, the infimum in~\eqref{eq:alpha_hat_implicit} is obtained when
\begin{equation}
    \frac{1}{\hat{\alpha}^{\mathrm{BCP}}(\mathbf{x},\mathcal{D})}
    =
    {E\left(\mathbf{x},y^{(C_{\max}(\mathbf{x})+1)}\right)},
\end{equation}
which gives~\eqref{eq:alpha_hat_closed_form}.

Therefore, we have the equivalence relationship
\begin{equation}
    y\notin\mathcal{C} (\mathbf{x})
    \Longleftrightarrow
    E(\mathbf{x},y)
    \ge
    \frac{1}{\hat{\alpha}^{\mathrm{BCP}}(\mathbf{x},\mathcal{D})}.
\end{equation}
For any data-dependent miscoverage level $\tilde{\alpha}>0$ that may depend on the calibration data $\mathcal{D}$ and the test input $\mathbf{x}$, the e-variable $E(\mathbf{x},y)$ satisfies the post-hoc validity property~\cite{koning2023post}
\begin{equation}
\label{eq:post_hoc_e}
    \mathbb{E}\left[
\frac{\mathds{1}\left(E(\mathbf{x},y)\ge 1/\tilde{\alpha}\right)}
    {\tilde{\alpha}}
    \right]\le 1 .
\end{equation}
Substituting $\tilde{\alpha}=\hat{\alpha}^{\mathrm{BCP}}(\mathbf{x},\mathcal{D})$ into~\eqref{eq:post_hoc_e} yields~\eqref{eq:post_hoc}, completing the proof.

\end{IEEEproof}

% Unlike the original BCP framework~\cite{gauthier2025backward,gauthier2025values}, where the per-input miscoverage probability is defined implicitly through the size constraint, we derive the estimate $\hat{\alpha}^{\mathrm{BCP}}(\mathbf{x},\mathcal{D})$ in closed form for the generalized top-$K$ set $\mathcal{C} (\mathbf{x})$ and prove its reliability in the sense of~\eqref{eq:cvg_guar_e}. 
% This guarantee provides a reliable assessment of the miscoverage risk for downstream mitigation.
Since the BCP estimate in~\eqref{eq:alpha_hat_closed_form} is positive, it lies in the interval $(0,1)$ if and only if \(s(\mathbf{x},y^{(C_{\max}(\mathbf{x})+1)})>N_{\mathrm{cal}}^{-1}\sum_{i=1}^{N_{\mathrm{cal}}}s(\mathbf{x}_i,y_i)\). This condition requires the first label excluded by the budget to have a larger score than the average true-label score in the calibration data. It is natural because it requires the budget to retain labels that are at least as plausible as a typical true label.

\section{Numerical Results}
\label{sec:numerical_results}

%--------------------------------------------------------------------
\begin{figure*}[t]
    \centering
    \subfloat[Ratio ($\mathrm{SIR}=0$ dB)]{%
        \includegraphics[width=0.25\textwidth]{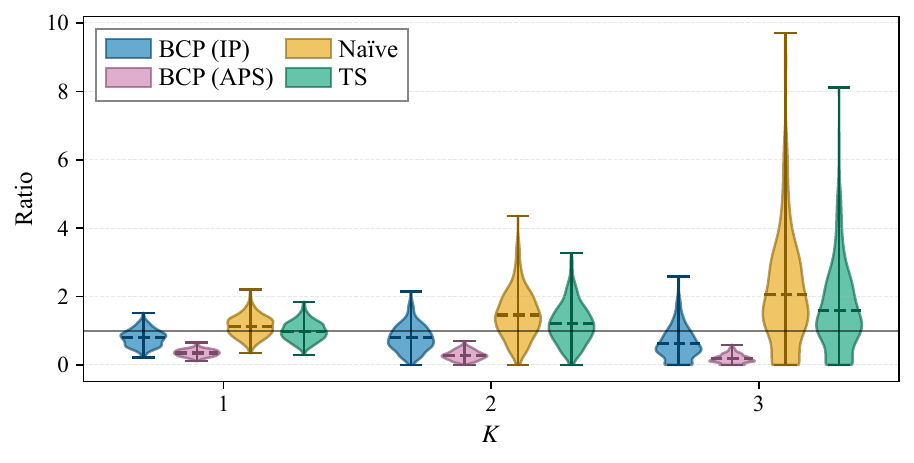}
        \label{fig:ratio_sir0}}
    % \hspace{0.02\textwidth}
    \subfloat[Difference ($\mathrm{SIR}=0$ dB)]{%
        \includegraphics[width=0.25\textwidth]{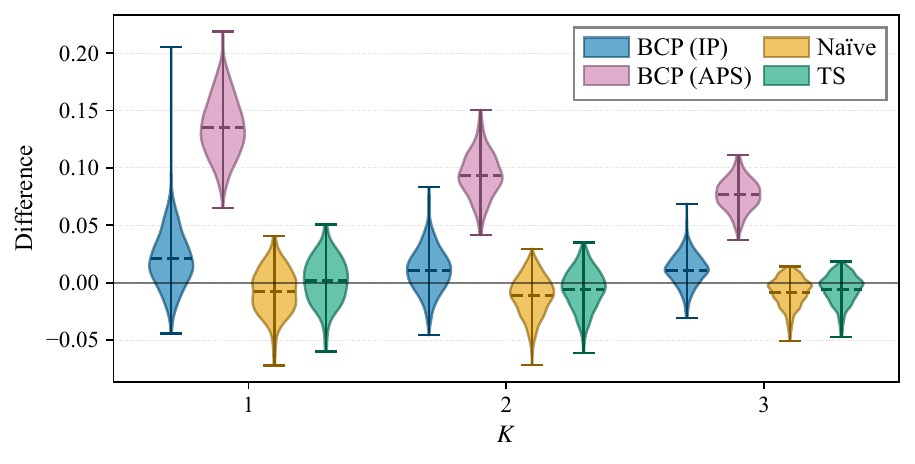}
        \label{fig:alphadiff_sir0}}
        % \hspace{0.02\textwidth}
    \subfloat[Ratio ($\mathrm{SIR}=5$ dB)]{%
        \includegraphics[width=0.25\textwidth]{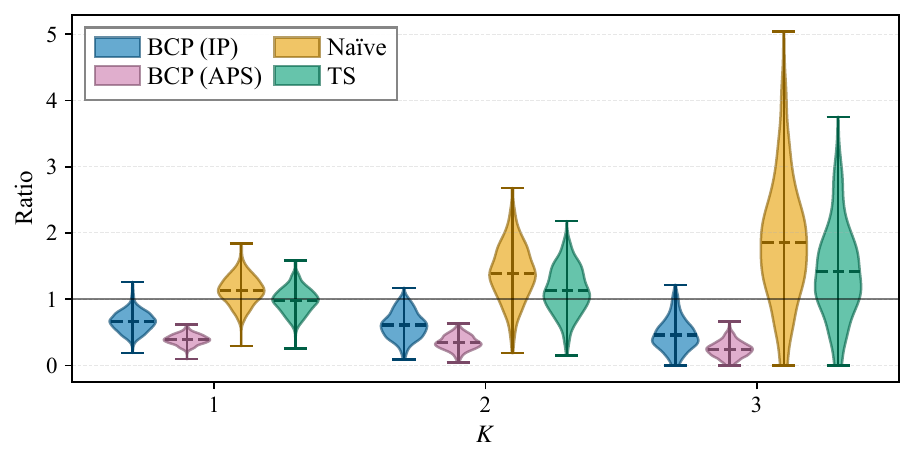}
        \label{fig:alphadiff_sir5}}
        % \hspace{0.02\textwidth}
    \subfloat[Difference ($\mathrm{SIR}=5$ dB)]{%
        \includegraphics[width=0.25\textwidth]{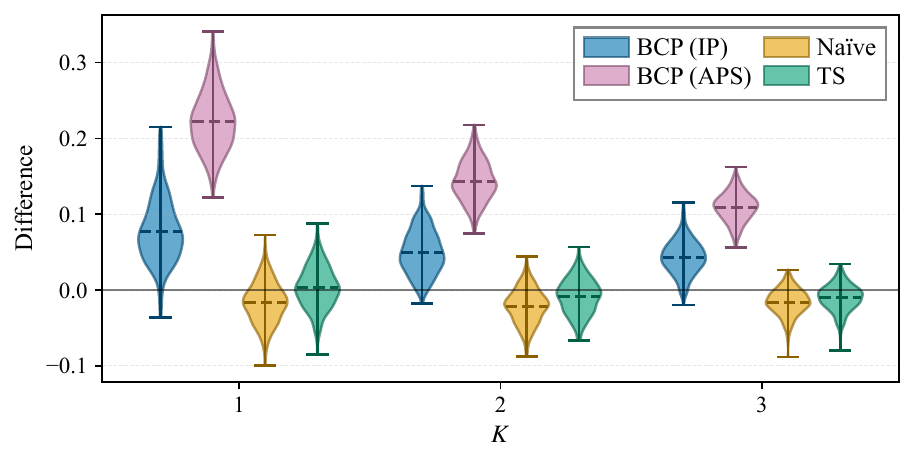}
        \label{fig:ratio_sir5}}
    \caption{Distribution of the reliability ratio~\eqref{eq:metric_rr} and signed miscoverage difference~\eqref{eq:metric_smd} at $\mathrm{SIR}\in\{0,5\}$ dB.
    % The dashed lines inside the violins indicate the corresponding average.
    }

    \label{fig:alphadiff}
\end{figure*}
% %--------------------------------------------------------------------
% \begin{figure}[t]
%     \centering
%     \subfloat[$\mathrm{SIR}=0$ dB]{%
%         \includegraphics[width=0.35\textwidth]{figures_v1/nbi_I64_B4_N4_SIRm10to10_brier_byK_bymethod_T1_T2_T3_n1000_SIR0_Ntest100_nout.pdf}
%         \label{fig:brier_sir0}}
%     \hspace{0.02\textwidth}
%     \subfloat[$\mathrm{SIR}=5$ dB]{%
%         \includegraphics[width=0.35\textwidth]{figures_v1/nbi_I64_B4_N4_SIRm10to10_brier_byK_bymethod_T1_T2_T3_n1000_SIR5_Ntest100_nout.pdf}
%         \label{fig:brier_sir5}}
%     \caption{Distribution of the Brier score~\eqref{eq:metric_bs} at $\mathrm{SIR}\in\{0,5\}$ dB. The dashed lines inside the violins indicate the corresponding average.}
%     \label{fig:brier}
% \end{figure}
% %--------------------------------------------------------------------
\subsection{Simulation Setup}
\label{subsec:sim_setup}
 
We validate the proposed NBI framework in a Wi-Fi system following the model in Sec.~\ref{subsec:sys_model}. Following~\cite{robinson2023narrowband}, the WiFi signal occupies a $20$~MHz channel and the narrowband interferer has a bandwidth of $156$~kHz. All I/Q data are generated according to the IEEE~802.11a/g standard using the MATLAB WLAN Toolbox~\cite{ieee802112020}. The training dataset contains approximately $120{,}000$ examples per label, with the signal-to-interference ratio (SIR) ranging from $-10$ to $10$~dB. Using the CNN architecture of~\cite{robinson2023narrowband} and
stochastic gradient descent, we train separate models for $S=4$ and
$S=16$ monitored subcarriers with observation lengths $M=64$ and
$M=128$, respectively.
% The detector adopts the CNN architecture of~\cite{robinson2023narrowband} and is trained via stochastic gradient descent. 
% For evaluation, we use a held-out dataset with $3{,}000$ examples per label at $\mathrm{SIR}\in\{0,5\}$~dB, and we compare the BCP miscoverage estimate (Sec.~\ref{subsec:bcp}) with the na\"{\i}ve baseline (Sec.~\ref{subsec:naive}).
For evaluation, we use a held-out dataset with $3{,}000$ examples per label at $\mathrm{SIR}\in\{0,5\}$~dB. 
We refer to BCP using the IP score in~\eqref{eq:nc_score_power} with
$\beta=0.75$ as BCP (IP). For comparison, we consider three
alternatives: the na\"{\i}ve baseline in Sec.~\ref{subsec:naive}, the NME with temperature scaling
(TS)~\cite{guo2017calibration}, and BCP with the adaptive prediction
sets (APS)-based score~\cite{romano2020classification}, which is defined as
\begin{equation}
    s^{\mathrm{APS}}(\mathbf{x},y)
    =
    -\log\left(
    1-\!\!\!\!\!\!\!\!
    \sum_{\substack{y'\in\mathcal Y:\\
    s(\mathbf{x},y')<s(\mathbf{x},y)}}\!\!\!\!\!\!
    \hat p(y'|\mathbf{x})
    -U\hat p(y|\mathbf{x})
    \right),
    \label{eq:aps_score}
\end{equation}
where $U\sim\operatorname{Unif}[0,1]$ is drawn independently for each
data point. We refer to the NME with temperature scaling as TS and to BCP with the
APS score as BCP (APS).
% \blue{In addition to the na\"{\i}ve baseline, we consider the NME with
% temperature scaling (TS)~\cite{guo2017calibration} and BCP with the adaptive prediction sets (APS)-based score~\cite{romano2020classification}}
% \begin{equation}
%     s^{\mathrm{APS}}(\mathbf{x},y)
%     =
%     -\log\left(
%     1-\!\!\!\!\!\!\!\!
%     \sum_{\substack{y'\in\mathcal Y:\\
%     s(\mathbf{x},y')<s(\mathbf{x},y)}}\!\!\!\!\!\!
%     \hat p(y'|\mathbf{x})
%     -U\hat p(y|\mathbf{x})
%     \right),
%     \label{eq:aps_score}
% \end{equation}
% where $U\sim\operatorname{Unif}[0,1]$ is drawn independently for each data point. We refer to BCP using the IP score in~\eqref{eq:nc_score_power} with $\beta=0.75$ and to the APS-based score in~\eqref{eq:aps_score} as BCP (IP) and BCP (APS), respectively.

\subsection{Performance Metrics}
\label{subsec:metrics}
%--------------------------------------------------------------------
Performance is evaluated on a test dataset $\mathcal{D}^{\mathrm{te}}=\{(\mathbf{x}_j,y_j)\}_{j=1}^{N_{\mathrm{te}}}$. For each test sample, let $\hat{\alpha}_j$ denote the miscoverage estimate and $m_j=\mathds{1}(y_j\notin\mathcal{C}(\mathbf{x}_j))$ the true miscoverage indicator. Based on these quantities, the estimated and true miscoverage rates are given by 
\begin{subequations}\label{eq:metric_mr}
\begin{align}
\text{Estimated miscoverage rate}
&=
\frac{1}{N_{\mathrm{te}}}
\sum_{j=1}^{N_{\mathrm{te}}}\hat{\alpha}_j,
\label{eq:metric_emr}\\
\text{True miscoverage rate}
&=
\frac{1}{N_{\mathrm{te}}}
\sum_{j=1}^{N_{\mathrm{te}}}m_j.
\label{eq:metric_tmr}
\end{align}
\end{subequations}
To assess the reliability of the miscoverage estimate, we report 
\begin{subequations}\label{eq:metric_reliability}
\begin{align}
\text{Reliability ratio}
&=
\frac{1}{N_{\mathrm{te}}}
\sum_{j=1}^{N_{\mathrm{te}}}
\frac{m_j}{\hat{\alpha}_j},
\label{eq:metric_rr}\\
\text{Signed miscoverage difference}
&=
\frac{1}{N_{\mathrm{te}}}
\sum_{j=1}^{N_{\mathrm{te}}}
(\hat{\alpha}_j-m_j).
\label{eq:metric_smd}
\end{align}
\end{subequations}
A reliability ratio no larger than one and a positive signed miscoverage difference indicate reliability in the senses of \eqref{eq:cvg_guar_e} and \eqref{eq:cvg_guar}, respectively.

To assess the accuracy of the miscoverage estimate, we report the Brier score
\begin{align}
\label{eq:metric_bs}
    \text{Brier~score}
    =
    \frac{1}{N_{\mathrm{te}}}
    \sum_{j=1}^{N_{\mathrm{te}}}
    \left(\hat{\alpha}_j-m_j\right)^2,
\end{align}
which measures the mean squared error between the estimated and true miscoverage. 
% All metrics are averaged over $N_{\mathrm{run}}=500$ independent experiments, each with an independently drawn calibration--test split $\{\mathcal{D}^{\mathrm{cal}},\mathcal{D}^{\mathrm{te}}\}$.
All metrics are averaged over $N_{\mathrm{run}}=500$ independent draws of disjoint calibration-test sets from the held-out dataset. In all violin plots, the dashed lines within each violin indicate the corresponding average.

\subsection{Performance Analysis}
\subsubsection{Equal-Cost Setting}

We begin with the equal-cost setting for $S=4$ monitored
subcarriers, where $c_y=1$ for all $y\in\mathcal{Y}$, so that
$C_{\max}(\mathbf{x})=K$ in~\eqref{eq:Cmax_intf}. We compare BCP (IP) and BCP (APS) with the na\"{\i}ve and TS
baselines at $\mathrm{SIR}\in\{0,5\}$ dB, with calibration size
$N_{\mathrm{cal}}=1000$ and test size $N_{\mathrm{te}}=100$. Fig.~\ref{fig:miscoverage_results} provides an intuitive comparison by showing the average estimated miscoverage rate~\eqref{eq:metric_emr} and its distribution for different $K$ and SIRs. Across all settings, the average estimates of both BCP-based methods exceed the corresponding average true miscoverage rates~\eqref{eq:metric_tmr}, whereas the na\"{\i}ve baseline systematically underestimates them, with its distribution falling entirely below the true rate for some budgets. The TS baseline reduces this bias and yields estimates closer to the true miscoverage rates, yet still underestimates them for some budgets.
%--------------------------------------------------------------------

To assess reliability, Fig.~\ref{fig:alphadiff} further reveals the distributions of the reliability ratio~\eqref{eq:metric_rr} and signed miscoverage difference~\eqref{eq:metric_smd}. The average ratios of both BCP methods remain below one, indicating reliability in the sense of~\eqref{eq:cvg_guar_e}, whereas those of the NME and TS can exceed one. A similar contrast appears in the signed differences, which are predominantly positive for both BCP methods, in line with~\eqref{eq:cvg_guar}, but negative for the NME and TS, consistent with Fig.~\ref{fig:miscoverage_results}. Among the BCP methods, BCP (APS) exhibits substantially larger positive differences because its cumulative-score structure weakens the distinction among individual label probabilities.
%--------------------------------------------------------------------
\begin{figure}[!h]
    \centering
    \subfloat[$\mathrm{SIR}=0$ dB]{%
        \includegraphics[width=0.5\columnwidth]{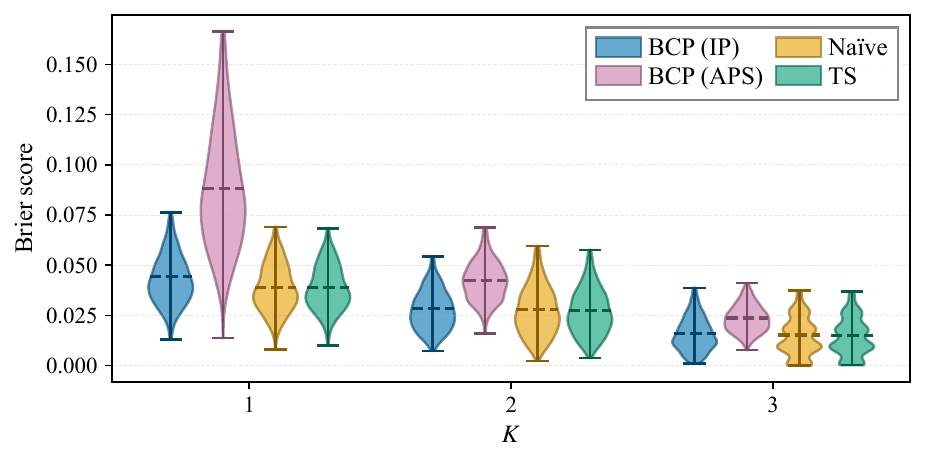}
        \label{fig:brier_sir0}}
    % \hspace{0.02\textwidth}
    \subfloat[$\mathrm{SIR}=5$ dB]{%
        \includegraphics[width=0.5\columnwidth]{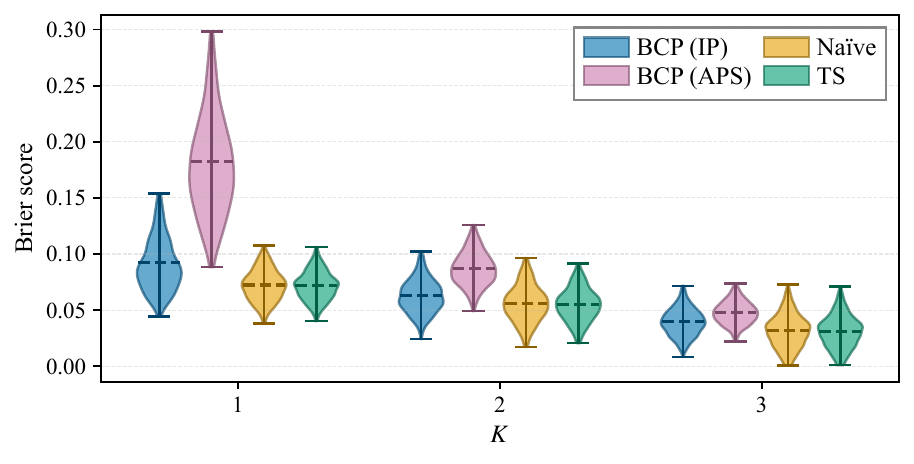}
        \label{fig:brier_sir5}}
    \caption{Distribution of the Brier score~\eqref{eq:metric_bs} at $\mathrm{SIR}\in\{0,5\}$ dB. 
    % The dashed lines inside the violins indicate the corresponding average.
    }
    \label{fig:brier}
\end{figure}
%--------------------------------------------------------------------
% %--------------------------------------------------------------------
% \begin{figure}[h]
%     \centering
%     \subfloat[$\mathrm{SIR}=0$ dB]{%
%         \includegraphics[width=0.48\columnwidth]{figures_v1/nbi_I64_B4_N4_SIRm10to10_brier_byK_bymethod_T1_T2_T3_n1000_SIR0_Ntest100_nout.pdf}
%         \label{fig:brier_sir0}}
%     \hfill
%     \subfloat[$\mathrm{SIR}=5$ dB]{%
%         \includegraphics[width=0.48\columnwidth]{figures_v1/nbi_I64_B4_N4_SIRm10to10_brier_byK_bymethod_T1_T2_T3_n1000_SIR5_Ntest100_nout.pdf}
%         \label{fig:brier_sir5}}
%     \caption{Distribution of the Brier score~\eqref{eq:metric_bs} at
%     $\mathrm{SIR}\in\{0,5\}$ dB. The dashed lines inside the violins
%     indicate the corresponding averages.}
%     \label{fig:brier}
% \end{figure}
% %--------------------------------------------------------------------

Moreover, Fig.~\ref{fig:brier} evaluates individual-estimate accuracy using the Brier score~\eqref{eq:metric_bs}, which is not captured by the magnitude in~\eqref{eq:metric_smd}. The NME and TS baselines yield nearly identical Brier scores. The proposed BCP (IP) achieves Brier scores comparable to these baselines, and substantially lower than BCP (APS). Together with Fig.~\ref{fig:alphadiff}, this shows that the proposed
scheme matches the accuracy of the NME and TS while providing reliable
miscoverage estimates.

\subsubsection{Unequal-Cost Setting}
%--------------------------------------------------------------------
\begin{figure*}[t]
    \centering
    \subfloat[Reliability ratio]{%
        \includegraphics[width=0.32\textwidth]{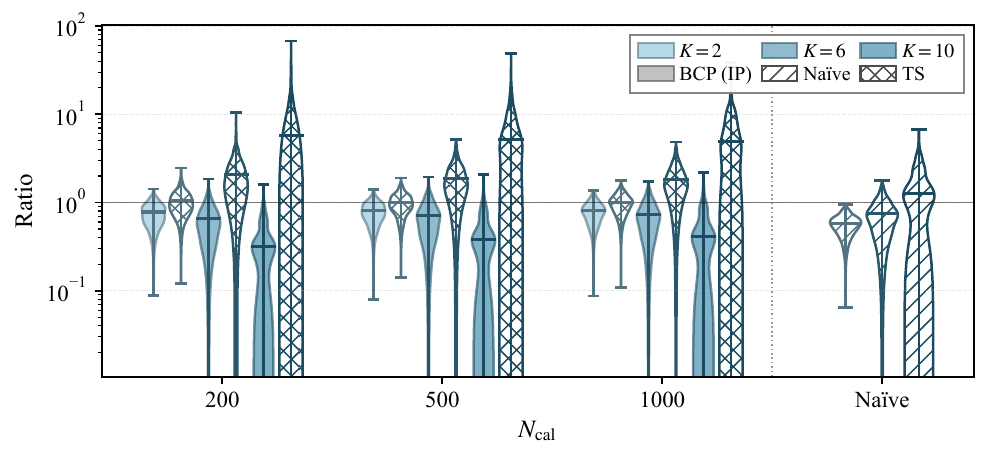}
        \label{fig:asym_ratio}}
    \hfill
    \subfloat[Signed miscoverage difference]{%
        \includegraphics[width=0.32\textwidth]{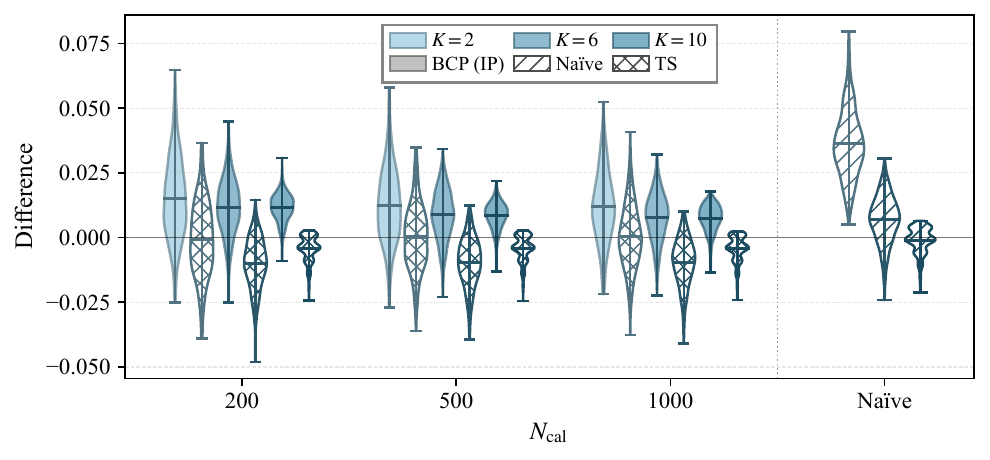}
        \label{fig:asym_alphadiff}}
    \hfill
    \subfloat[Brier score]{%
        \includegraphics[width=0.32\textwidth]{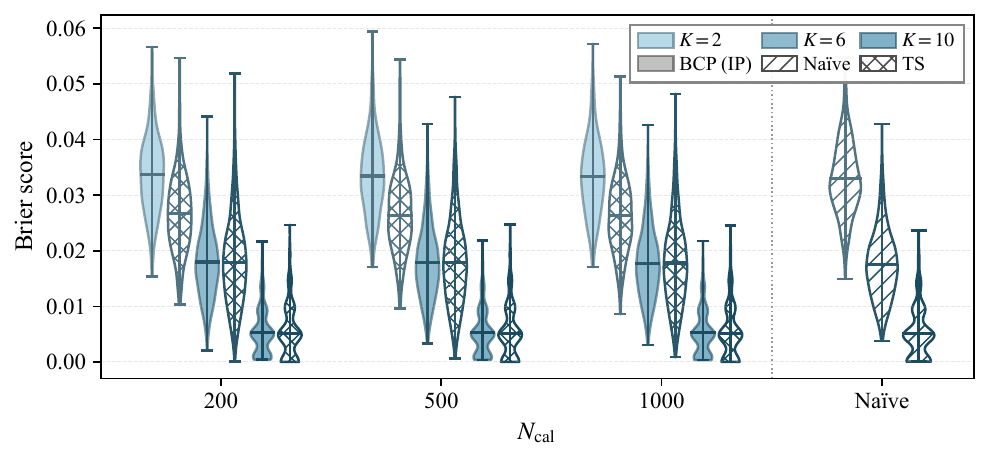}
        \label{fig:asym_brier}}
    \caption{Distribution of the reliability ratio~\eqref{eq:metric_rr},
    signed miscoverage difference~\eqref{eq:metric_smd}, and Brier
    score~\eqref{eq:metric_bs} versus the calibration size
    $N_{\mathrm{cal}}$ at $\mathrm{SIR}=5$ dB. 
    % The dashed lines inside
    % the violins indicate the corresponding averages.
    }
    \label{fig:asym_metrics}
\end{figure*}
%--------------------------------------------------------------------
We next consider the unequal-cost setting for $S=16$ monitored
subcarriers at $\mathrm{SIR}=5$ dB. Specifically, we set $c_y=1$ for the first eight subcarriers and
$c_y=0.5$ for the remaining eight. We compare BCP (IP) with the na\"{\i}ve and TS baselines for
$N_{\mathrm{cal}}\in\{200,500,1000\}$ and
$N_{\mathrm{te}}=200$. BCP (APS) is omitted due to its lower accuracy in the equal-cost setting.

%--------------------------------------------------------------------
% \begin{figure}[!ht]
%     \centering
%     \includegraphics[width=0.8\columnwidth]{figures_v1/nbi_I64_B4_N4_SIRm10to10_alphadiff_byNcal_byK_IP_naiveRef_K1_K1.5_K2_K2.5_n200_n500_n1000_SIR5_Ntest200_out.pdf}
%     \caption{Distribution of the signed miscoverage difference~\eqref{eq:metric_smd} versus the calibration size $N_{\mathrm{cal}}$ at $\mathrm{SIR}=5$ dB. The dashed lines inside the violins indicate the corresponding average.}
%     \label{fig:asym_alphadiff}
% \end{figure}
% %--------------------------------------------------------------------
% Fig.~\ref{fig:asym_metrics} evaluates the reliability and accuracy under unequal costs for different budgets $K$ and calibration sizes $N_{\mathrm{cal}}$. 
% Fig.~\ref{fig:asym_ratio} shows that the average reliability ratios of BCP remain below one for all considered budgets and calibration sizes, indicating reliability in the sense of~\eqref{eq:cvg_guar_e}, whereas those of the NME consistently exceed one. Similarly, Fig.~\ref{fig:asym_alphadiff} shows predominantly positive signed differences for BCP, in line with~\eqref{eq:cvg_guar}, while the NME yields negative differences. As $N_{\mathrm{cal}}$ increases, both BCP distributions become more concentrated, indicating more stable estimates with larger calibration sets. Finally, Fig.~\ref{fig:asym_brier} shows that the Brier score of BCP decreases with $N_{\mathrm{cal}}$ and approaches that of the NME, indicating that the accuracy cost of conservativeness diminishes with more calibration data.
Fig.~\ref{fig:asym_metrics} evaluates the reliability and accuracy
under unequal costs for different budgets $K$ and calibration sizes
$N_{\mathrm{cal}}$. Fig.~\ref{fig:asym_ratio} shows that the average reliability ratios
of BCP remain below one for all settings. By contrast, those of the
NME can exceed one, while TS yields substantially larger ratios. Similarly,
Fig.~\ref{fig:asym_alphadiff} shows predominantly positive signed
differences for BCP, while the NME can yield negative differences
and those of TS are predominantly negative. As $N_{\mathrm{cal}}$ increases, the BCP and TS distributions
become more concentrated, indicating more stable estimates.
Nevertheless, the reliability violations of TS persist because its
single global temperature does not directly calibrate the miscoverage
of cost- and budget-dependent prediction sets.
Finally,
Fig.~\ref{fig:asym_brier} shows that BCP, the NME, and TS achieve
comparable Brier scores, which remain stable across
$N_{\mathrm{cal}}$ and decrease with $K$. 
Overall, the proposed BCP scheme maintains reliability under unequal
costs while achieving high accuracy.

%--------------------------------------------------------------------
% \begin{figure}[!ht]
%     \centering
%     \includegraphics[width=0.8\columnwidth]{figures_v1/nbi_I64_B4_N4_SIRm10to10_brier_byNcal_byK_IP_naiveRef_K1_K1.5_K2_K2.5_n200_n500_n1000_SIR5_Ntest200_out.pdf}
%     \caption{Distribution of the Brier score~\eqref{eq:metric_bs} versus the calibration size $N_{\mathrm{cal}}$ at $\mathrm{SIR}=5$ dB. The dashed lines inside the violins indicate the corresponding average.}
%     \label{fig:asym_brier}
% \end{figure}
% %--------------------------------------------------------------------
% Fig.~\ref{fig:asym_brier} further examines the accuracy in terms of the Brier score. The Brier score of BCP decreases with $N_{\mathrm{cal}}$ and approaches that of the NME, indicating that the accuracy cost of conservativeness diminishes with more calibration data.

\section{Conclusion}
\label{sec:conclusion}
This paper proposed a BCP-based framework for budget-constrained narrowband interference detection in wireless systems. Operationally constrained prediction sets are first constructed over interference subcarriers from the output of a probabilistic detector, with their size capped by a limited mitigation budget. BCP is then used to assess, for each input, the risk that the selected set fails to contain the true interference state, while preserving distribution-free reliability guarantees. 
% Numerical results have validated both the reliability and the accuracy of the proposed method, showing that it consistently yields conservative miscoverage estimates whose Brier-score gap to a na\"{\i}ve baseline shrinks as the calibration size grows.
Numerical results validate both the reliability and accuracy of
the proposed method, showing that it consistently yields conservative
miscoverage estimates while achieving estimation accuracy comparable
to that of the NME and TS baselines.
Future work may consider extensions to broader wireless tasks and to stronger reliability guarantees beyond the marginal setting considered here.

% \appendix[Derivation of \eqref{eq:cvg_guar}]
% \bluec{
% Using the definition of covariance, we have
% \begin{align}
% \mathbb{E}\!\left[\hat{\alpha}(\mathbf{x},\!\mathcal{D})\right]
% \!-\!\Pr\!\left(y\!\notin\!\mathcal{C}(\mathbf{x})\right)
% \!&=\!
% \mathbb{E}\!\left[\hat{\alpha}(\mathbf{x},\!\mathcal{D})\right]
% \!\left(\!1\!-\!
% \mathbb{E}\!\left[\!
% \frac{\mathds{1}(y\!\notin\!\mathcal{C}(\mathbf{x}))}
% {\hat{\alpha}(\mathbf{x},\mathcal{D})}\!
% \right]\!\right)
% \nonumber\\
% &-
% \operatorname{Cov}\!\left(
% \frac{\mathds{1}(y\notin\mathcal{C}(\mathbf{x}))}
% {\hat{\alpha}(\mathbf{x},\mathcal{D})},
% \hat{\alpha}(\mathbf{x},\mathcal{D})
% \right).
% \label{eq:marginal_derivation}
% \end{align}
% Given~\eqref{eq:cvg_guar_e}, the first term on the right-hand side
% of~\eqref{eq:marginal_derivation} is nonnegative. Moreover, by the law of
% iterated expectations, the covariance in the second term can be written as
% $\operatorname{Cov}\!\left(
% \Pr(y\notin\mathcal{C}(\mathbf{x})\mid\hat{\alpha})/\hat{\alpha},
% \hat{\alpha}\right)$.
% If this covariance is nonpositive, then~\eqref{eq:cvg_guar} holds.
% }

\appendix[Derivation of the Marginal 
Bound~\eqref{eq:cvg_guar_bound}]
By convexity of the ratio $1/\hat{\alpha}(\mathbf{x},\mathcal{D})$ for
$\hat{\alpha}(\mathbf{x},\mathcal{D})>0$, its tangent at
$\mathbb{E}[\hat{\alpha}(\mathbf{x},\mathcal{D})]$ gives
\begin{align}
\frac{1}{\hat{\alpha}(\mathbf{x},\mathcal{D})}
\geq
\frac{1}{\mathbb{E}[\hat{\alpha}(\mathbf{x},\mathcal{D})]}
-
\frac{
\hat{\alpha}(\mathbf{x},\mathcal{D})
-\mathbb{E}[\hat{\alpha}(\mathbf{x},\mathcal{D})]
}{
\mathbb{E}^{2}[\hat{\alpha}(\mathbf{x},\mathcal{D})]
}.
\end{align}
Multiplying both sides by
$\mathds{1}(y\notin\mathcal{C}(\mathbf{x}))$ and taking expectations yield
\begin{align}
\mathbb{E}\!\left[\!\frac{\mathds{1}(y\!\notin\!\mathcal{C}(\mathbf{x}))}{\hat{\alpha}(\mathbf{x},\mathcal{D})}\!\right]
\!\!\geq\!\!
\frac{\Pr(y\!\notin\!\mathcal{C}(\mathbf{x})\!)}{\mathbb{E}[\hat{\alpha}(\mathbf{x},\mathcal{D})]}
\!-\!
\frac{\operatorname{Cov}\!\left(\mathds{1}(y\notin\mathcal{C}(\mathbf{x})),\!\hat{\alpha}(\mathbf{x},\mathcal{D})\right)}
{\mathbb{E}^{2}[\hat{\alpha}(\mathbf{x},\mathcal{D})]}.
\label{eq:mul_exp}
\end{align}
Combining~\eqref{eq:mul_exp} with the ratio guarantee
in~\eqref{eq:cvg_guar_e} gives
\begin{align}
\Pr(y\notin\mathcal{C}(\mathbf{x}))
\leq
\mathbb{E}[\hat{\alpha}(\mathbf{x},\mathcal{D})]
+
\frac{
\operatorname{Cov}\!\left(
\mathds{1}(y\notin\mathcal{C}(\mathbf{x})),
\hat{\alpha}(\mathbf{x},\mathcal{D})
\right)
}{
\mathbb{E}[\hat{\alpha}(\mathbf{x},\mathcal{D})]
}.
\label{eq:cvg_cov_bound}
\end{align}
Since
$\operatorname{Var}(\mathds{1}(y\notin\mathcal{C}(\mathbf{x})))\leq 1/4$
for any indicator variable, the Cauchy-Schwarz inequality gives
\begin{align}
\operatorname{Cov}\!\left(
\mathds{1}(y\notin\mathcal{C}(\mathbf{x})),
\hat{\alpha}(\mathbf{x},\mathcal{D})
\right)
\leq
\frac{1}{2}
\sqrt{\operatorname{Var}\!\left(
\hat{\alpha}(\mathbf{x},\mathcal{D})\right)}.
\label{eq:cauchy_schwarz}
\end{align}
Substituting~\eqref{eq:cauchy_schwarz} into~\eqref{eq:cvg_cov_bound} gives~\eqref{eq:cvg_guar_bound}.

\bibliographystyle{IEEEtran}
\bibliography{./bibtex/bib/IEEEabrv,./bibtex/bib/IEEEreference}

\end{document}